\documentclass{article}
\usepackage{fullpage}

\usepackage{xspace}
\usepackage{amsmath,amssymb,amsthm}
\usepackage{algorithmicx}
\usepackage[noend]{algpseudocode}
\algrenewcomment[1]{\hfill \(\slash\slash\) {\tt#1}}
\usepackage{forloop}
\newcounter{ct}
\newcommand{\markdent}[1]{\forloop{ct}{0}{\value{ct} < #1}{\hspace{\algorithmicindent}}}
\newcommand{\markcomment}[1]{\Statex\markdent{#1}}
\usepackage{multirow}
\usepackage{multicol}
\usepackage{tikz}
\def\checkmark{\tikz\fill[scale=0.4](0,.35) -- (.25,0) -- (1,.7) -- (.25,.15) -- cycle;} 

\newcommand{\ADN}{{Anonymous Dynamic Network}\xspace}
\newcommand{\ADNs}{{Anonymous Dynamic Networks}\xspace}

\newcommand{\name}{\textsc{Methodical Counting}\xspace}

\newcommand{\mig}[1]{\textcolor{blue}{#1}}
\newcommand{\dk}[1]{\textcolor{red}{#1}}
\renewcommand{\mig}[1]{#1}
\renewcommand{\dk}[1]{#1}
\renewcommand{\vec}[1]{\mathbf{#1}}
\newtheorem{theorem}{Theorem}
\newtheorem{lemma}{Lemma}
\newtheorem{corollary}{Corollary}

\newtheorem{claim}{Claim}
\begin{document}

%\setcounter{chapter}{2} % If you are doing your chapter as chapter one,
%\setcounter{section}{3} % comment these two lines out.

%\title{\Large SIAM/ACM Preprint Series Macros for
%Use With LaTeX\thanks{Supported by GSF grants ABC123, DEF456, and GHI789.}}
\title{Polynomial Counting in\\ \ADNs\\\vspace{.1in}\large{with Applications to 
\dk{Anonymous Dynamic} Algebraic Computations}}

%\author{Corey Gray\thanks{Society for Industrial and Applied Mathematics.} \\
%\and
%Tricia Manning\thanks{Society for Industrial and Applied Mathematics.}}
\author{
Dariusz~R.~Kowalski 
\thanks{Computer Science Department, 
University of Liverpool, 
Liverpool, UK.
E-mail: D.Kowalski@liverpool.ac.uk}\\
\and~Miguel~A.~Mosteiro 
\thanks{Computer Science Department, 
Pace University,
New York, NY, USA.
E-mail: mmosteiro@pace.edu}
}

\date{}

\maketitle

% Copyright Statement
% When submitting your final paper to a SIAM proceedings, it is requested that you include 
% the appropriate copyright in the footer of the paper.  The copyright added should be 
% consistent with the copyright selected on the copyright form submitted with the paper.
% Please note that "20XX" should be changed to the year of the meeting.

% Default Copyright Statement
%\fancyfoot[R]{\footnotesize{\textbf{Copyright \textcopyright\ 20XX by SIAM\\
%Unauthorized reproduction of this article is prohibited}}}

% Depending on which copyright you agree to when you sign the copyright form, the copyright 
% can be changed to one of the following after commenting out the default copyright statement
% above.

%\fancyfoot[R]{\footnotesize{\textbf{Copyright \textcopyright\ 20XX\\
%Copyright for this paper is retained by authors}}}

%\fancyfoot[R]{\footnotesize{\textbf{Copyright \textcopyright\ 20XX\\
%Copyright retained by principal author's organization}}}

%\pagenumbering{arabic}
%\setcounter{page}{1}%Leave this line commented out.

\begin{abstract} \small\baselineskip=9pt 
%This is the text of my abstract. It is a brief description of my
%paper, outlining the purposes and goals I am trying to address.

%!TEX root = ./KM_counting_SIAM.tex

Starting with Michail, Chatzigiannakis, and Spirakis work~\cite{spirakis}, the problem of \emph{Counting} the number of nodes in \ADNs has attracted a lot of attention. 
The problem is challenging because nodes are indistinguishable (they lack identifiers and execute the same program) and the topology may change arbitrarily from round to round of communication, as long as the network is connected in each round. 
The problem is central in distributed computing as the number of participants is frequently needed to make important decisions, such as termination, agreement, synchronization,
and many others.
A variety of algorithms built on top of \emph{mass-distribution} techniques have been presented, analyzed, and also experimentally evaluated;
some of them assumed additional knowledge of network characteristics,
such as bounded degree or given upper bound on the network size. 
However, the question of whether Counting can be solved deterministically in sub-exponential time remained open. 
In this work, we answer this question positively by presenting \name, which runs in polynomial time and requires no knowledge of network characteristics. 
\mig{Moreover, we also show how to extend \name to compute the sum of input values and more complex functions without extra cost.}
Our analysis leverages previous work on random walks in evolving graphs, \mig{combined with carefully chosen alarms in the algorithm that control the process and its parameters}. 
To the best of our knowledge, our Counting algorithm and its extensions to other algebraic and Boolean functions are the first that can be implemented in practice with worst-case guarantees. 
\end{abstract}

%\IEEEraisesectionheading{\section{Introduction}\label{sec:introduction}}
%% Computer Society journal (but not conference!) papers do something unusual
%% with the very first section heading (almost always called "Introduction").
%% They place it ABOVE the main text! IEEEtran.cls does not automatically do
%% this for you, but you can achieve this effect with the provided
%% \IEEEraisesectionheading{} command. Note the need to keep any \label that
%% is to refer to the section immediately after \section in the above as
%% \IEEEraisesectionheading puts \section within a raised box.

%!TEX root = ./KM_counting_SIAM.tex

\section{Introduction}\label{sec:introduction}

In this work, we address the \mig{standing} question of whether the number of nodes of an  \ADN (ADN) can be counted deterministically in polynomial time or not.
We answer this question positively by presenting the \name algorithm, and proving formally that after a polynomial number of rounds of communication all nodes know the size of the network and stop. 

The problem has been thoroughly studied~\cite{spirakis,conscious,oracle,experimentalConscious,LunaB15,opodisCounting,netysCounting} because Counting 
is central for distributed computing. Indeed, more complex tasks need the network size to make various decisions on state agreement, synchronization, termination, and others.
%decide termination. 
However, \ADNs pose a particularly challenging scenario. On one hand, nodes are indistinguishable from each other. For instance, they may lack identifiers or their number may be so massive that keeping record of them is not feasible. On the other hand, the topology of the network is highly dynamic. Indeed, the subsets of nodes that may communicate with each other may change all the time. 
All these features make ADN a valid model for anonymous ad hoc
communication and computation.

In such a restrictive scenario, finding a way of providing theoretical guarantees of deterministic polynomial time has been elusive until now. Indeed, previous papers have either 
weaken the objective (e.g., computing only upper bound, only stochastic guarantees, etc.), 
assumed availability of network information (e.g., maximum number of neighbors, size upper bound, etc.), 
relied on a stronger model of communication,
or provided only superpolynomial time guarantees.

\name uses no information about the network. After completing its execution, all nodes obtain the exact size of the network and stop. Moreover, they stop all at the same time, allowing the algorithm to be concatenated with other computations.

Our algorithm is based on nodes continuously sharing some magnitude, which we call \emph{potential},\footnote{In previous related works this quantity, used in
a different way, was termed \emph{energy}. We steer away from such denomination to avoid confusion with node energy supply.} resembling \emph{mass-distribution} and \emph{push-pull} algorithms. 
Unlike previous algorithms, in \name carefully and periodically (i.e. , ``methodically'') some potential is removed from the network, rather than greedily doing so continuously. 
This approach is combined with another methodological innovation testing whether the candidate
value (for the network size) is within some polynomial range of the actual network size.
This complex strategy yields an algorithm 
\dk{in which the progress in mass-distribution} 
%that 
can be 
analyzed 
%similarly 
%%%modeled 
as a sequence of \dk{parametrized} Markov chains (even though the algorithm itself is purely deterministic) \dk{enhanced by mass drift and alarms controlling the process
and its parameters}. 
Our analysis approach opens the path to study more complex tasks  in \ADNs \mig{applying similar techniques.}

\mig{
Finally, we also present a variety of extensions of \name to compute more complex functions. Most notably, we present an extension that, concurrently with finding the network size, computes the sum of input values held at each node without asymptotic time overhead. Having a method to compute the sum and network size, more complex computations are possible in polynomial time as well. Indeed, we also describe how to compute a variety of algebraic and Boolean functions.
To the best of our knowledge, ours are the first algorithms for anonymous dynamic Counting and other algebraic computations that can be implemented in practice with worst-case guarantees. 
}

\subsubsection*{Roadmap:} The rest of the paper is organized as follows. We specify the model and notation details in Section~\ref{prelim}. Then, we overview previous work in Section~\ref{relwork} and present our results in Section~\ref{results}. Section~\ref{algorithm} includes the details of \name, and we prove its correctness and running time in Section~\ref{analysis}. Extensions to other functions are presented in Section~\ref{s:extensions}. %and improvements to \name are described in Section~\ref{conclude}.

\section{Model, Problem, and Notation}
\label{prelim}

\subsubsection*{The Counting Problem:}
The definition of the problem is simple. An algorithm solves the Counting Problem if, after completing its execution, all nodes have obtained the exact size of the network and stop.

\subsubsection*{\ADNs:}
The following model is customary in the \ADNs literature.
We consider a network composed by a set $V$ of $n>1$ network \emph{nodes} with processing and communication capabilities. 
It was shown in~\cite{spirakis} that Counting cannot be solved in Anonymous Networks without the availability of at least one distinguished node in the network.
Thus, we assume the presence of such node called \emph{leader}.
Aside from the leader, we assume that all other nodes are indistinguishable from each other. That is, we do not assume the availability of labels or identifiers, and all non-leader nodes execute exactly the same program.

Each pair of nodes that are able to communicate define a communication \emph{link}, and the set of links is called the \emph{topology} of the network. 
The nodes in a communication link are called \emph{neighbors}.
The event of sending a message to neighbors is called a \emph{broadcast} or \emph{transmission}. 
Nodes and links are reliable, in the sense that no communication or node failures occur.
Hence, a broadcasted message is received by all neighbors.
Moreover, links are \emph{symmetric}, that is, if node $a$ is able to send a message to node $b$, then $b$ is able to send a message to $a$.

Without loss of generality, we discretize time in \emph{rounds}.
In any given round, a node may broadcast a message, receive all messages from broadcasting neighbors, and carry out some computations, in that order. The time taken by the computations is assumed to be negligible. 

The set of links among nodes may change from round to round, and nodes have no way of knowing which were the neighbors they had before.
These topology changes are arbitrary, limited only to maintain the network connected
in each round. That is, at any given round the topology is such that there is a \emph{path}, i.e., a sequence of links, between each pair of nodes, but the set of links may change arbitrarily from round to round. This adversarial model of dynamics was called \emph{$1$-interval connectivity} in~\cite{KuhnLO2010}.

The following notation will be used. 
The maximum number of neighbors that any node may have at any given time is called the \emph{dynamic maximum degree} and it is denoted as $\Delta$. 
The maximum length of a path between any pair of nodes at any given time is called the \emph{dynamic diameter} and it is denoted as $D$. 
The maximum length of an opportunistic path between any pair of nodes over many time slots is called the \emph{chronopath}~\cite{FCFMMZ:randomgeocast} and it is denoted as $\mathcal{D}$.

%!TEX root = ./KM_counting_SIAM.tex

\begin{table*}[htbp]
\centering
\footnotesize
\begin{tabular}{|c|c|c|c|c|c|c|}
\hline
\rule{0pt}{4ex}
\multirow{2}{*}{algorithm}&\multicolumn{2}{c|}{needs}&\multirow{2}{*}{computes}&\multirow{2}{*}{stops?}&\multicolumn{2}{c|}{complexity}\\
[.1in]
\cline{2-3}
\cline{6-7}
\rule{0pt}{4ex}
&\begin{tabular}{c}size\\ upper\\ bound\\ $N$\end{tabular}&\begin{tabular}{c}dynamic\\ maximum\\ degree u.b.\\$d_{\max}$\end{tabular}&&&time&space\\
[.1in]
\hline
\hline
\rule{0pt}{4ex}
\begin{tabular}{c}\emph{Degree}\\\emph{Counting}~\cite{spirakis}\end{tabular}&&\checkmark&$O(d_{\max}^n)$&\checkmark&$O(n)$&\\
[.1in]
\hline
\rule{0pt}{4ex}
\emph{Conscious}~\cite{conscious}&\checkmark&\checkmark&$n$&\checkmark&\begin{tabular}{c}$O(e^{N^2}N^3) \Rightarrow$\\ $O(e^{d_{\max}^{2n}}d_{\max}^{3n})$ using~\cite{spirakis}\end{tabular}&\\
[.1in]
\hline
\rule{0pt}{4ex}
\emph{Unconscious}~\cite{conscious}&&&$n$&No&\begin{tabular}{c}No theoretical\\ bounds\end{tabular}&\\
[.1in]
\hline
\rule{0pt}{4ex}
$\mathcal{A}_{\mathcal{O}^P}$~\cite{oracle}&&\begin{tabular}{c}Degree oracle\\ for each node\end{tabular}&$n$&Eventually&Unknown&\\
[.1in]
\hline
\rule{0pt}{4ex}
\textsf{EXT}~\cite{LunaB15}&&&$n$&\checkmark&$O(n^{n+4})$&EXPSPACE\\
[.1in]
\hline
\rule{0pt}{4ex}
\begin{tabular}{c}\textsc{Incremental}\\\textsc{Counting}~\cite{opodisCounting}\end{tabular}&&\checkmark&$n$&\checkmark&$O\left(n\left(2d_{\max}\right)^{n+1}\frac{\ln n}{\ln d_{\max}}\right)$&\\
[.1in]
\hline
\hline
\rule{0pt}{4ex}
\begin{tabular}{c}\textsc{Methodical}\\ \textsc{Counting}\\$[$This work$]$\end{tabular}&&&$n$&\checkmark&$O(n^5\ln^2 n)$&PSPACE\\
[.1in]
\hline
\end{tabular}
\caption{Comparison of Counting protocols for \ADNs.}
\label{table}
\end{table*}

%!TEX root = ./KM_counting_SIAM.tex

\section{Previous Work}
\label{relwork}
In this section we overview previous work directly related to this paper. A comprehensive overview of work related to \ADNs can be found in a survey by Casteigts et al.~\cite{arnaudSurvey} and references in the papers cited here.
The related work overviewed, in comparison with our results, is summarized in Table~\ref{table}.

With respect to lower bounds, it was proved in~\cite{baldoni} that at least $\Omega(\log n)$ rounds are needed, even if $D$ is constant. Also, a trivial observation is that $\Omega(\mathcal{D})$ is a lower bound as at least one node needs to hear about all other nodes to obtain the right count, and the chronopath $\mathcal{D}$ is the largest number of hops that a message from some node needs to take to reach other node in the network, possibly along multiple time slots.

Counting was already studied in~\cite{spirakis}, together with the problem of \emph{Naming}, for dynamic and static networks. It was shown in this work that it is impossible to solve Counting without the presence of a distinguished node, even if nodes do not move. The Counting protocol presented for \ADNs requires knowledge of an upper bound on $\Delta$, and the count obtained is only an upper bound on the network size, which may be as bad as exponential. 

An exact count is obtained by the Conscious Counting algorithm presented in~\cite{conscious}. However, the computation relies on knowing initially an upper bound on the network size. The running time of this protocol is exponential only if the initial upper bound is tight. 

In the same work and follow-up papers~\cite{oracle,experimentalConscious}, the authors presented protocols under more challenging scenarios where $\Delta$ is not known. However, either the protocol does not terminate~\cite{conscious}, and hence the running time cannot be bounded, or the protocol is terminated heuristically~\cite{experimentalConscious}. In experiments~\cite{experimentalConscious}, such heuristic was found to perform well on dense topologies, but for other topologies the error rate was high. That is, the results only apply to dense \ADNs. Another protocol in~\cite{oracle} is shown to terminate eventually, without running-time guarantees and under the assumption of having for each node an estimate of the number of neighbors in each round. In~\cite{spirakis} it was conjectured that some knowledge of the network  such as the latter would be necessary, but the conjecture was disproved later in~\cite{LunaB15}. On the other hand the protocol in~\cite{LunaB15} requires exponential space. 

Recently, a protocol called Incremental Counting was presented in~\cite{opodisCounting}. This algorithm reduced exponentially the running time guarantees with respect to previous works developed under the same model. Incremental Counting obtains the exact count, all nodes terminate simultaneously, the topology dynamics is only limited to $1$-interval connectivity, it only requires polynomial space, and it only requires knowledge of the dynamic maximum degree $\Delta$. The superpolynomial running time proved still does not provide enough guarantee for practical application, but reducing from doubly-exponential to exponential was an important step towards understanding the complexity of Counting. 

In a follow-up paper~\cite{netysCounting}, Incremental Counting was tested experimentally showing a promising polynomial behavior. The study was conducted on pessimistic inputs designed to slow the convergence, such as bounded-degree trees rooted at the leader uniformly chosen at random for each round, and a single path starting at the leader with non-leader nodes permuted uniformly at random for each round. The protocol was also tested on static versions of the inputs mentioned, classic random graphs, and networks where some disconnection is allowed. The results exposed important observations. Indeed, even for topologies that stretch the dynamic diameter, the running times obtained are below $\Delta n^3$. It was also observed that random graphs, as used in previous experimental studies~\cite{experimentalConscious}, reduce the convergence time, and therefore are not a good choice to indicate worst-case behavior. These experiments showed good behavior even for networks that sometimes are disconnected, indicating that more relaxed models of dynamics, such as ($\alpha,\beta$)-connectivity~\cite{FCFMMZ:randomgeocast,geocast}, are worth to study. All in all, the experiments in~\cite{netysCounting} showed that Incremental Counting behaves well in a variety of pessimistic inputs, but not having a proof of what a worst-case input looks like, and being the experiments restricted to a range of values of $n$ far from the expected massive size of an \ADN, a theoretical proof of polynomial time remained an open problem even from a practical perspective.

In a recent manuscript~\cite{BaldoniTR} a polynomial Counting algorithm is presented relying on the availability of an algorithm to compute average with polynomial convergence time. Such average computation is modeled as a Markov chain with underlying doubly-stochastic matrix, which requires topology information within two hops (cf.~\cite{nedic2009}). In our model of \ADN, such information is not available,
and gathering it may not be possible due to possible topology changes from round to
round.

%!TEX root = ./KM_counting_SIAM.tex

\section{Our Contributions}
\label{results}

We present and analyze a deterministic distributed algorithm to compute the number of nodes in an \ADN. We call such algorithm \name. As opposed to previous works, our algorithm does not require any knowledge of network characteristics, such as dynamic maximum degree or an upper bound on the size. 
After $O(n^5\ln^2 n)$ communication rounds of running \name, all nodes obtain the network size and stop at the same round. 
To the best of our knowledge, this is the first polynomial deterministic Counting algorithm
in the pure model of \ADN.

Our algorithm is based on distributing potential in a mass-distribution fashion, similarly as previous works for Counting. The main algorithmic novelty in our approach is that the leader participates in the process as any other node, removing potential only after it has accumulated enough. This approach allowed us to leverage previous work on random walks in evolving graphs. For this approach to work, we combine it with testing whether the candidate value for the network size is polynomially close to the actual value.
Our approach also opens the path to study more complex computations in \ADNs using the same analysis.

\mig{Finally, we also present extensions of \name to compute more complex functions. Most notably, we show how to modify \name to compute the sum of input values held by nodes at the same time than counting. Having an algorithm to compute the network size and the sum of input values, we also show how to compute other algebraic and Boolean functions.}

%!TEX root = ./KM_counting_SIAM.tex

\section{\name}
\label{algorithm}

\begin{figure*}
\hrule
\caption{\name algorithm for the leader. $N$ is the set of neighbors of the leader in the current round. The parameters $d,p,r$ and $\tau$ are as defined in Theorem~\ref{thm}.}
\label{leaderAlg}
\vspace{.1in}
\hrule
\begin{algorithmic}[1]
\Procedure{Count}{}%{$a,b$}
	\State $\rho\gets 0$ \Comment{accumulator of consumed potential}
	\State $\Phi\gets 0$ \Comment{current potential}
	\State $k\gets 2$ \Comment{current estimate}
	\State $status\gets normal$ \Comment{status$=$normal$|$alarm$|$done}
	\While{$status\neq done$}   \Comment{\dk{iterating} epochs} \label{epochs}
		\For{$phase=1$ to $p$} \Comment{\dk{iterating} phases} \label{phases}
			\For{$round=1$ to $r$} \Comment{\dk{iterating} rounds} \label{rounds}
				\State Broadcast $\langle\Phi,status\rangle$
				and Receive $\langle\Phi_i,status_i\rangle, \forall i\in N$
				\If{$status=normal$ {\bf and} $|N|\leq d-1$ {\bf and} $\forall i\in N:status_i=normal$} 
				%\Comment{update potential}	
					\State $\Phi\gets \Phi + \sum_{i\in N}\Phi_i/d - |N|\Phi/d$ 
					\Comment{\dk{update potential}}	
					\label{potupdate}
				\Else \Comment{$k$ is wrong} \label{leadertoomany}
					\State $status\gets alarm$\label{alarminsecondleader}
					\State $\Phi\gets 1$
				\EndIf
			\EndFor 
			\markcomment{3}{{\tt\slash* $r$ rounds completed *\slash}}
			\If{$phase=1$ {\bf and} $\Phi> \tau$} \Comment{$k$ is wrong} \label{leaderthreshold}
					\State $status\gets alarm$
					\State $\Phi\gets 1$
			\EndIf 
			\If{$status=normal$}	\Comment{prepare for next phase}
				\State $\rho \gets \rho + \Phi$ \label{rhoupdate}
				\State $\Phi \gets 0$ \label{phireset}
			\EndIf
		\EndFor 
		\markcomment{2}{{\tt\slash* $p$ phases completed *\slash}}
		\If{$status=normal$ {\bf and} $k-1-1/k\leq \rho\leq k-1$} \Comment{the size is $k$} \label{range}
			\State $status\gets done$
		\Else \Comment{prepare for next epoch} \label{reset}
			\State $\rho\gets 0$
			\State $\Phi\gets 0$
			\State $k\gets k+1$
			\State $status\gets normal$ 
		\EndIf
		\For{$round=1$ to $k$} \Comment{disseminate termination} \label{leadernotification}
			\State Broadcast $\langle status\rangle$
			and Receive $\langle status_i\rangle, \forall i\in N$
		\EndFor 
		\markcomment{2}{{\tt\slash* epoch completed *\slash}}
	\EndWhile 
	\State \textbf{return} $k$ 
\EndProcedure
\end{algorithmic}
\hrule
\end{figure*}

\begin{figure*}
\hrule
\caption{\name algorithm for each non-leader node $i$. $N$ is the set of neighbors of $i$ in the current round. The parameters $d,p,r$ and $\tau$ are as defined in Theorem~\ref{thm}.}
\label{otherAlg}
\vspace{.1in}
\hrule
\begin{algorithmic}[1]
\Procedure{Count}{}
	\State $\Phi\gets 0$ \Comment{current potential}
	\State $k\gets 2$ \Comment{current estimate}
	\State $status\gets normal$ \Comment{status$=$normal$|$alarm$|$done}
	\While{$status\neq done$}   \Comment{\dk{iterating} epochs} \label{epochs}
		\For{$phase=1$ to $p$} \Comment{\dk{iterating} phases} \label{phases}
			\For{$round=1$ to $r$} \Comment{\dk{iterating} rounds} \label{rounds}
				\State Broadcast $\langle\Phi,status\rangle$
				and Receive $\langle\Phi_i,status_i\rangle, \forall i\in N$
				\If{$status=normal$ {\bf and} $|N|\leq d-1$ {\bf and} $\forall i\in N:status_i=normal$} 
				%\Comment{update potential}	
					\State $\Phi\gets \Phi + \sum_{i\in N}\Phi_i/d - |N|\Phi/d$ 
					\Comment{update potential}
					\label{newpot}
				\Else \Comment{$k$ is wrong} \label{othertoomany}
					\State $status\gets alarm$\label{alarminsecondother}
					\State $\Phi\gets 1$
				\EndIf
			\EndFor 
			\markcomment{3}{{\tt\slash* $r$ rounds completed *\slash}}
			\If{$phase=1$ {\bf and} $\Phi> \tau$} \Comment{$k$ is wrong} \label{otherthreshold}
					\State $status\gets alarm$ \label{thresholdalarm}
					\State $\Phi\gets 1$
			\EndIf 
		\EndFor 
		\markcomment{2}{{\tt\slash* $p$ phases completed *\slash}}
		\For{$round=1$ to $k$} \Comment{disseminate termination} \label{othernotification}
			\State Broadcast $\langle status\rangle$
			and Receive $\langle status_i\rangle, \forall i\in N$
			\If{$\exists i\in N:status_i=done$}
				\State $status\gets done$
			\EndIf
		\EndFor 
		\If{$status\neq done$}
			\State $k\gets k+1$
			\State $status\gets normal$
		\EndIf
		\markcomment{2}{{\tt\slash* epoch completed *\slash}}
	\EndWhile 
	\State \textbf{return} $k$ 
\EndProcedure
\end{algorithmic}
\hrule
\end{figure*}

In this section we present \name. First, we give the intuition of the algorithm, the details can be found in Figures~\ref{leaderAlg} and~\ref{otherAlg}. (References to algorithm lines are given as $\langle figure\#\rangle.\langle line\#\rangle$.)

Initially, the leader is assigned a potential of $0$ and all the other nodes are assigned a potential of~$1$. 
Then, the algorithm is composed by epochs, each of which is divided into phases composed by rounds of communication. Epoch $k$ corresponds to a size estimate $k$ that is iteratively increased from epoch to epoch until the correct value $n$ is found.
Each epoch is divided into $p$ phases. The purpose of each phase is for the leader to collect as much potential as possible from the other nodes in a mass-distribution fashion as follows.

Each phase is composed by $r$ rounds of communication. In each round, each node\footnote{As opposed to previous work, in \name the leader also follows this procedure.} broadcasts its potential and receives the potential of all its neighbors. Each node keeps only a fraction $1/d$ of the potentials received. The parameters $p$, $r$, and $d$ are functions of $k$. 
The specific functions needed to guarantee correctness \dk{and saught efficiency} are defined in Theorem~\ref{thm}. 
\dk{This varying way of distributing potential is different from previous approaches
using mass distribution.} 
After communication, each node updates its own potential accordingly (cf. Lines~\ref{leaderAlg}.\ref{potupdate} and~\ref{otherAlg}.\ref{newpot}). That is, it adds a fraction $1/d$ of the potentials received, and subtracts a fraction $1/d$ of the potential broadcasted times the number of potentials received. Then, a new round starts. 

At the end of each phase the leader ``consumes'' its potential. That is, it increases an internal accumulator $\rho$ with its current potential, which is zeroed for \dk{starting} the next phase (cf. Lines~\ref{leaderAlg}.\ref{rhoupdate} and~\ref{leaderAlg}.\ref{phireset}). 
A node stops the update of potential described, raises its potential to $1$, and broadcasts an alarm in each round until the end of the epoch if any of the following happens: 1) at the end of the first phase its potential is above some threshold $\tau$ as defined in Theorem~\ref{thm} (cf. Lines~\ref{leaderAlg}.\ref{leaderthreshold} and~\ref{otherAlg}.\ref{otherthreshold}), 2) at any round it receives more than $d-1$ messages
\dk{(cf. Lines~\ref{leaderAlg}.\ref{leadertoomany} and~\ref{otherAlg}.\ref{othertoomany})}, or 3) at any round it receives an alarm (cf. Lines~\ref{leaderAlg}.\ref{leadertoomany} and~\ref{otherAlg}.\ref{othertoomany}). 
The alarm for case 1) allows the leader to detect that the estimate is wrong when $k^{1+\epsilon}<n$ for some $\epsilon>0$ (Lemmas~\ref{unalarmed} and~\ref{alarmsoon}), the alarm for case 2) allows the leader to detect that $d$ is too small and hence the estimate is wrong, and the alarm for case 3) allows dissemination of all alarms.
In the alarm status the potential is set to $1$ to facilitate the analysis, but it is not strictly needed by the algorithm.

At the end of each epoch, the leader checks the value of $\rho$. If $k-1-1/k\leq \rho\leq k-1$ the current estimate is correct and the leader changes its status to ``done'' (cf. Line~\ref{leaderAlg}.\ref{range}). Otherwise, all its variables are reset to start a new epoch with the next estimate (cf. Line~\ref{leaderAlg}.\ref{reset}).
Before starting a new epoch the network is flooded with the status of the leader for $k$ rounds (cf. Lines~\ref{leaderAlg}.\ref{leadernotification} and~\ref{otherAlg}.\ref{othernotification}). If $k=n$, \dk{the leader initiates message ``done''} and the $k$ rounds are enough for all the nodes to receive the ``done'' status and after completing the $k$ rounds stop. Otherwise, nodes will not receive the ``done'' status and after completing the $k$ rounds they start a new epoch.

%!TEX root = ./KM_counting_SIAM.tex

\section{Analysis}
\label{analysis}

%\subsubsection*{Notation:} 

In this section we analyze \name. 
References to algorithm lines are given as $\langle figure\#\rangle.\langle line\#\rangle$.
\mig{
We will use the standard notation for the $L_p$ norm of vector $\vec{x}=(x_1,x_2,\dots,x_n)$ as $||\vec{x}||_p = \left(\sum_{i=1}^n |x_i|^p\right)^{1/p}$, for any $p\geq 1$.
}
Only for the analysis, nodes are labeled as $0,1,2,\dots,n-1$, where the leader has label $0$.
The potential of a node $i$ at the beginning of round $t$ is denoted as $\Phi_{t}[i]$, and the potential of all nodes at the beginning of round $t$ is denoted as a vector $\vec{\Phi}_t$. The aggregated potential is then $||\vec{\Phi}_t||_1$. 
The subindex $t$ is used for rounds, phases, or dropped as needed.
We will refer to the potential right after the last round of a phase as $\vec{\Phi}_{r+1}$. Such round does not exist in the algorithm, but we use this notation to distinguish between the potential right before the leader consumes its own potential (cf. Line~\ref{leaderAlg}.\ref{reset}) and the potential at the beginning of the first round of the next phase.

First, we provide a broad description of our analysis of \name. 
%(Refer to the algorithm in Figures~\ref{leaderAlg} and~\ref{otherAlg} as needed.)
Consider the vector of potentials $\vec{\Phi}_i$ held by nodes at the beginning of any given phase $i$.
The way that potentials are updated in each round (cf. Lines~\ref{leaderAlg}.\ref{potupdate} and~\ref{otherAlg}.\ref{newpot}) is equivalent to the progression of a $d$-lazy random walk on the evolving graph underlying the network topology~\cite{michal}, where the initial vector of potentials is equivalent to an initial distribution $\vec{p}_i$ on the overall potential $||\vec{\Phi}_i||_1$
\dk{and the probability of choosing a specific neighbor is $1/d$}. 
For instance, the initial vector of potentials $\vec{\Phi}_0=\langle0,1,1,\dots\rangle$, corresponds to a distribution $\vec{p}_0=\langle 0,1/(n-1),1/(n-1),\dots\rangle$ on the initial $||\vec{\Phi}_0||_1=n-1$.

\dk{Note that our \name\ is not a simple ``derandomization'' of the lazy
random walk on evolving graphs. 
First, in the \ADN\ model neighbors cannot be 
distinguished, and even their number is unknown at transmission time
(only at receiving time the node learns the number of its neighbors). 
Second, due to unknown network parameters,
it may happen in an execution of \name\ that the total potential received
could be bigger than $1$. Third, 
%due to similar reasons, 
our algorithm does not know a priori when to terminate and provide result
even with some reasonable accuracy, as the formulas on mixing and cover time
of lazy random walks depend on (a priori unknown) number of nodes $n$. 
Nevertheless, we can still use some results obtained in the context of 
analogous lazy random walks
in order to prove useful properties of parts of algorithm \name,
namely, some parts in which parameters are temporarily fixed and
the number of received messages does not exceed parameter $d$.}

It was shown in~\cite{michal} that random walks on $d$-regular explorable evolving graphs have a uniform stationary distribution, and bounds on the mixing and cover time were proved as well. Moreover, it was observed that those properties hold even if the graph is not regular and $d$ is only an upper bound on the degree.\footnote{Their analysis relies on Lemma 12, which bounds the eigenvalues of the transition matrix as long as it is stochastic, connected, symmetric, and non-zero entries lower bounded by $1/d$. Those conditions hold for all the transition matrices, even if the evolving graph is not regular.}

Thus, for the cases where $d$ is an upper bound on the number of neighboring nodes, we analyze the evolution of potentials within each phase leveraging previous work on random walks on evolving graphs. Specifically, we use the following result which is an extension of Corollary 14 in~\cite{michal}.
\begin{theorem}
\label{koucky}
(Corollary 14 in~\cite{michal}.)
After $t$ rounds of a $d_{\max}$-lazy random walk on an evolving graph with $n$ nodes, dynamic diameter $D$, upper bound on maximum degree $ d_{\max}$, and initial distribution $\vec{p}_0$, the following holds.
\begin{align*}
\left|\left|\vec{p}_t - \frac{\vec{I}}{n}\right|\right|_2^2 \leq \left(1-\frac{1}{ d_{\max}Dn}\right)^t\left|\left|\vec{p}_0 - \frac{\vec{I}}{n}\right|\right|_2^2
\end{align*}
\end{theorem}
%In our notation, $1/d_{\max}$ is the fraction of potential $1/d$ shared by each node, and $h$ is the dynamic diameter $D$.

In between phases the leader ``consumes'' its potential, effectively changing the distribution at that point. Then, a new phase starts. 

In \name, given that $d$ is a function of the estimate $k$, if the estimate is low there may be inputs for which $d$ is not an upper bound on the number of neighbors. We show in our analysis that in those cases the leader detects the error and after some time all nodes increase the estimate.

%For clarity, we analyze \name using $d=k^2$, which allows a clean separation of the cases $k^2\geq n$ and $k^2< n$. The running time obtained in this way is polynomial on $n$. Using $d=k$, the algorithm obtains the network size faster, but the analysis is more laborious. The details are included in Section~\ref{conclude}.

First, we prove correctness when $k=n$ as follows.

\begin{lemma}
\label{correct}
If $d\geq k$ and $k=n$, after running the \name protocol for $p\geq\frac{k}{1-1/k}\ln (k(k-1))$ phases, each of $r\geq4dk^2\ln k$ rounds, the potential $\rho$ consumed by the leader is $k-1-1/k \leq \rho \leq k-1$.
\end{lemma}

\begin{proof}
The second inequality is immediate because the initial total potential in the network is $n-1$ and it does not increase during the execution. So, if $k=n$, the potential consumed by the leader cannot be more than $k-1$.

For the first inequality, consider the vector of potentials $\vec{\Phi}_1$ at the beginning of round $1$ of any phase $i$.
As explained above, we analyze the evolution of potentials within phase $i$ as a random walk on the evolving graph underlying the network topology.
%
%The evolution of the algorithm within a phase can be viewed as a random walk, where the initial vector of potentials is equivalent to an 
Consider the initial distribution $\vec{p}_i$ on the overall potential $||\vec{\Phi}_1||_1$. 
%For instance, the initial vector of potentials $\vec{\Phi}_0=\langle0,1,1,\dots\rangle$, corresponds to a distribution $\vec{p}_0=\langle 0,1/(k-1),1/(k-1),\dots\rangle$ on the initial $||\vec{\Phi}_0||_1=k-1$.
%
Then, using Theorem~\ref{koucky}, we know that after a phase $i$ of $r\geq4dk^2\ln k$ rounds the distribution is such that 
\begin{align}
\left|\left|\vec{p}_{r+1} - \frac{\vec{I}}{k}\right|\right|_2^2 
&\leq \left(1-\frac{1}{d{\cal D}k}\right)^r\left|\left|\vec{p}_1 - \frac{\vec{I}}{k}\right|\right|_2^2\label{distance}\\
&\leq \exp\left(-\frac{r}{d{\cal D}k}\right)\nonumber\\
&\leq \exp\left(-\frac{4dk^2\ln k}{d{\cal D}k}\right),\textrm{ given that $k=n>{\cal D}$,}\nonumber\\
&\leq \exp\left(-4\ln k\right)\nonumber\\
&= \frac{1}{k^4}.\nonumber
\end{align}
Given that $(p_{r+1}[0] - 1/k)^2 \leq \left|\left|\vec{p}_{r+1} - \frac{\vec{I}}{k}\right|\right|_2^2$, we have that $(p_{r+1}[0]-1/k)^2 \leq 1/k^4$ and hence $p_{r+1}[0] \geq 1/k - 1/k^2$.
Notice that the latter is true for any initial distribution, as the distance to uniform in Equation~\ref{distance} has been upper bounded by $1$.
Thus, applying recursively we have that after $p\geq\frac{k}{1-1/k}\ln (k(k-1))$ phases it is
\begin{align*}
||\vec{\Phi}_p||_1 
&\leq \left(1-\frac{1}{k}\left(1-\frac{1}{k}\right)\right)^p (k-1)\\
&\leq \exp\left(-\frac{p}{k}\left(1-\frac{1}{k}\right)\right) (k-1)\\
&\leq  1/k.
\end{align*}
Thus, the claim follows.
\hfill$\square$
\end{proof}

%%%%%%%%%%%%%%%%%%%%%%%%%%%%%%%%%%%%%%%%%%%%%%%%%%%%%%%%%%%%%%%%%%%%%%%%%%%

The previous lemma shows that if $\rho>k-1$ or $\rho<k-1-1/k$ we know that the estimate $k$ is wrong, but the complementary case, that is, $k-1-1/k \leq \rho \leq k-1$, may occur even if the estimate is $k<n$ and hence the error has to be detected by other means. To prove correctness in that case, we show first that if $k<n\leq k^{1+\epsilon}$ for some $\epsilon>0$ the leader must consume $\rho>k-1$ potential if the protocol is run long enough. To ensure that $d\geq  \Delta+1$, we restrict $d\geq k^{1+\epsilon}$.

\begin{lemma}
\label{ksquare}
If $1<k<n\leq k^{1+\epsilon}\leq d$, $\epsilon>0$, after running the \name protocol for $p\geq\frac{(2+\epsilon)k^{1+\epsilon}}{1-1/k}\ln k$ phases, each of $r\geq(4+2\epsilon)dk^{2+2\epsilon}\ln k$ rounds, the potential $\rho$ consumed by the leader is $\rho > k-1$.
\end{lemma}

\begin{proof}
Given that $d\geq n$, we can use Theorem~\ref{koucky} as in Lemma~\ref{correct} to show that after a phase $i$ of $r\geq(4+2\epsilon)dk^{2+2\epsilon}\ln k$ rounds the distribution is such that 
\begin{align}
\left|\left|\vec{p}_{r+1} - \frac{\vec{I}}{n}\right|\right|_2^2 
&\leq \left(1-\frac{1}{d{\cal D}n}\right)^r\left|\left|\vec{p}_1 - \frac{\vec{I}}{n}\right|\right|_2^2\nonumber\\
&\leq \exp\left(-\frac{r}{d{\cal D}n}\right)\nonumber\\
&\leq \exp\left(-\frac{(4+2\epsilon)dk^{2+2\epsilon}\ln k}{d{\cal D}n}\right) \ ,\textrm{ given that $k^{1+\epsilon}\geq n > {\cal D}$,}\nonumber\\
&\leq \exp\left(-(4+2\epsilon)\ln k\right)\nonumber\\
&= \frac{1}{k^{4+2\epsilon}} \ .\nonumber
\end{align}

Given that $(p_{r+1}[0] - 1/n)^2 \leq \left|\left|\vec{p}_{r+1} - \frac{\vec{I}}{n}\right|\right|_2^2$, we have that $(p_{r+1}[0]-1/n)^2 \leq 1/k^{4+2\epsilon}$ and hence $p_{r+1}[0] \geq 1/n - 1/k^{2+\epsilon}$.
The latter is true for any initial distribution, as the distance to uniform has been upper bounded by $1$.
So, applying recursively, we have that after $p\geq\frac{(2+\epsilon)k^{1+\epsilon}}{1-1/k}\ln k$ phases it is
\begin{align*}
||\vec{\Phi}_p||_1 
&\leq \left(1-\left(\frac{1}{n}-\frac{1}{k^{2+\epsilon}}\right)\right)^p (n-1)\\
&\leq \exp\left(-p\left(\frac{1}{n}-\frac{1}{k^{2+\epsilon}}\right)\right) (n-1) \ , \textrm{ since $k^{1+\epsilon}\geq n$,}\\
&\leq \exp\left(-\frac{p}{k^{1+\epsilon}}\left(1-\frac{1}{k}\right)\right) (n-1) \ , \textrm{ replacing $p$,}\\
%&\leq \exp\left(-\frac{\frac{(2+\epsilon)k^{1+\epsilon}}{1-1/k}\ln k}{k^{1+\epsilon}}\left(1-\frac{1}{k}\right)\right) (n-1), \\
%&= \exp\left(-\frac{(2+\epsilon)}{1-1/k}\ln k\left(1-\frac{1}{k}\right)\right) (n-1), \\
%&= \exp\left(-(2+\epsilon)\ln k\right) (n-1), \\
&\leq \frac{n-1}{k^{2+\epsilon}} \ , \textrm{ given that $k^{1+\epsilon}>n-1$,}\\
&<  1/k \ .
\end{align*}
Thus, the potential consumed by the leader is $\rho\geq n-1-1/k > k-1$ for any integers $n>k>1$.
\hfill$\square$
\end{proof}

%%%%%%%%%%%%%%%%%%%%%%%%%%%%%%%%%%%%%%%%%%%%%%%%%%%%%%%%%%%%%%%%%%%%%%%%%%%

It remains to show that even if $n>k^{1+\epsilon}$ \name still detects that the estimate is low. 
%We focus on the first phase. We define a threshold $\tau$ such that, after the phase is completed, all nodes that have potential above $\tau$ can send an alarm to the leader, as such potential indicates that the estimate is low. We show that the alarm must be received after $k^{1+\epsilon}$ further rounds of communication. 
First, we prove the following two claims that establish properties of the potential during the execution of \name.
%First, we show that if the number of neighbors of each node is less than $d$ in every round, the overall potential is neither lost nor created. 
(Recall that we use round $r+1$ to refer to potentials at the end of the phase right before the leader consumes its potential in Line~\ref{leaderAlg}.\ref{reset}.)

\begin{claim}
\label{conservation}
Given an \ADN of $n$ nodes running \name with parameter $d$, for any round $t$ of the first phase, such that $1\leq t\leq r+1$, if $d$ was larger than the number of neighbors of each node $x$ for every round $t'<t$, then $||\vec{\Phi}_t||_1=n-1$. 
\end{claim}

\begin{proof}
For the first round the claim holds as the initial potential of each node is $1$ except the leader that gets $0$. That is, $||\vec{\Phi}_1||_1 = n-1$.
For any given round $1< t\leq r+1$ in phase $1$ and any given node $x$, if $d$ is larger than the number of neighbors of $x$, the potential is updated only in Lines~\ref{leaderAlg}.\ref{potupdate} and~\ref{otherAlg}.\ref{newpot} as
\begin{align*}
\Phi_{t+1}[x] &= \Phi_{t}[x] + \sum_{i\in N_{t}[x]}\Phi_{t}[i]/d - |N_{t}[x]|\Phi_{t}[x]/d
\ .
\end{align*}
Where 
%$\Phi_{i}[j]$ is the potential of node $j$ at the beginning of round $i$, and 
$N_{t}[x]$ is the set of neighbors of node $x$ in round $t$.
Inductively, assume that the claim holds for some round $1\leq t\leq r$. 
We want to show that consequently it holds for $t+1$.
The potential for round $t+1$ is
\begin{align}
||\vec{\Phi}_{t+1}||_1 &= ||\vec{\Phi}_{t}||_1 + \frac{1}{d}\sum_{x\in V} \left( \sum_{y\in N_{t}[x]}\Phi_{t}[y] - |N_{t}[x]|\Phi_{t}[x] \right) \ .\label{potvecupdate}
\end{align}

In the \ADN model, communication is symmetric. That is, for every pair of nodes $x,y\in V$ and round $t$, it is $x\in N_{t}[y] \iff y\in N_{t}[x]$. 
Fix a pair of nodes $x',y' \in V$ such that in round $t$ it is $y'\in N_{t}[x']$ and hence $x'\in N_{t}[y']$. 
Consider the summations in Equation~\ref{potvecupdate}.
Due to symmetric communication, we have that the potential $\Phi_{t}[y']$ appears with positive sign when the indeces of the summations are $x=x'$ and $y=y'$, and with negative sign when the indices are $x=y'$ and $y=x'$. This observation applies to all pairs of nodes that communicate in any round $t$. 
Therefore, we can re-write Equation~\ref{potvecupdate} as
\begin{align*}
||\vec{\Phi}_{t+1}||_1 &= ||\vec{\Phi}_{t}||_1 + \frac{1}{d}\sum_{\substack{x,y\in V:\\y\in N_{t}[x]\\}} \bigg(\Phi_{t}[y] - \Phi_{t}[x] + \Phi_{t}[x] - \Phi_{t}[y] \bigg) = ||\vec{\Phi}_{t}||_1 \ .
\end{align*}
Thus, the claim follows.

\hfill$\square$
\end{proof}

\begin{claim}
\label{potbounds}
Given an \ADN of $n$ nodes running \name, for any round $t$ of any phase and any node $x$, it is $0\leq \Phi_t[x]\leq 1$.  
\end{claim}
\begin{proof}
If $t=1$ the potential of the leader is $\Phi_1[0]=0$ and the potential of any non-leader node $x$ is $\Phi_1[x]=1$. Thus, the claim follows. 
Inductively, for any round $2<t\leq r+1$, we consider two cases according to node status. 
If a node $x$ is in alarm status at the beginning of round $t$, then it is $\Phi_t[x]=1$ as, whenever the status of a node is updated to alarm, its potential is set to $1$ and will not change until the next epoch (cf. Figures~\ref{leaderAlg} and~\ref{otherAlg}).
On the other hand, if a node $x$ is in normal status at the beginning of round $t$, it had its potential updated in all rounds $t'<t$ only in Lines~\ref{leaderAlg}.\ref{potupdate} and~\ref{otherAlg}.\ref{newpot} as
\begin{align*}
\Phi_{t'+1}[x] &= \Phi_{t'}[x] + \sum_{y\in N_{t'}[x]}\Phi_{t'}[y]/d - |N_{t'}[x]|\Phi_{t'}[x]/d.
\end{align*}
For all rounds $t'<t$, node $x$ exchanged potential with less than $d$ neighbors, because otherwise it would have been changed to alarm status in Lines~\ref{leaderAlg}.\ref{alarminsecondleader} and~\ref{otherAlg}.\ref{alarminsecondother}.
Therefore it is $|N_{t'}[x]|\Phi_{t'}[x]/d < \Phi_{t'}[x]$ which implies $\Phi_t[x]\geq 0$. 
It can also be seen that $\Phi_t[x]\leq 1$ because, for any $t'<t$, it is
\begin{align*}
\Phi_{t'+1}[x] &= \Phi_{t'}[x] + \sum_{y\in N_{t'}[x]}\Phi_{t'}[y]/d - |N_{t'}[x]|\Phi_{t'}[x]/d\\
&\leq \Phi_{t'}[x] + \frac{|N_{t'}[x]|}{d} - \frac{|N_{t'}[x]|}{d}\Phi_{t'}[x]\\
&= \Phi_{t'}[x] + \frac{|N_{t'}[x]|}{d}(1 - \Phi_{t'}[x])\\
&\leq \Phi_{t'}[x] + 1 - \Phi_{t'}[x] =1.
\end{align*}
\hfill$\square$
\end{proof}

It remains to show that even if $n>k^{1+\epsilon}$ \name still detects that the estimate is low. 
We focus on the first phase. We define a threshold $\tau$ such that, after the phase is completed, all nodes that have potential above $\tau$ can send an alarm to the leader, as such potential indicates that the estimate is low. We show that the alarm must be received after $k^{1+\epsilon}$ further rounds of communication. 

\begin{lemma}
\label{unalarmed}
For $\epsilon>0$,
after running the first phase of the \name protocol, 
there are at most $k^{1+\epsilon}$ nodes that 
have potential at most $\tau=1-1/k^{1+\epsilon}$.
\end{lemma}

\begin{proof}
We define the \emph{slack} of node $x$ at the beginning of round $t$ as $s_t[x]=1-\Phi_t[x]$ and the vector of slacks at the beginning of round $t$ as $\vec{s}_t$. In words, the slack of a node is the ``room'' for additional potential up to $1$. 
Recall that the overall potential at the beginning of round $1$ of phase $1$ is $||\vec{\Phi}_1||_1=n-1$. 
Also notice that for any round and any node $x$ the potential of $x$ is non-negative as shown in Claim~\ref{potbounds}.
Therefore, the overall slack with respect to the maximum potential that could be held by all the $n$ nodes at the beginning of round $1$ is $||\vec{s}_1||_1=1$.

Consider a partition of the set of nodes $\{L,H\}$, where $L$ is the set of nodes with potential at most $\tau=1-1/k^{1+\epsilon}$ at the end of the first phase, before the leader consumes its own potential in Line~\ref{leaderAlg}.\ref{reset}. That is, $\Phi_{r+1}[x] \leq \tau$ for all $x\in L$.
Assume that the slack held by nodes in $L$ at the end of the first phase is at most the overall slack at the beginning of the phase. That is, $\sum_{x\in L}s_{r+1}[x] \leq ||\vec{s}_1||_1 = 1$. 
By definition of $L$, we have that for each node $x\in L$ it is $s_{r+1}[x]=(1-\Phi_{r+1}[x])\geq 1-\tau$.
Therefore,
$|L|(1-\tau) \leq \sum_{x\in L} s_{r+1}[x] \leq 1$.
Thus, $|L| \leq 1/(1-\tau) = k^{1+\epsilon}$ and the claim follows.

Then, to complete the proof, it remains to show that $\sum_{x\in L}s_{r+1}[x]\leq 1$.
Let the scenario where $d$ is larger than the number of neighbors that each node has in each round of the first phase be called ``case 1'', and ``case 2'' otherwise.
Claim~\ref{conservation} shows that in case 1 at the end of the first phase it is $||\vec{\Phi}_{r+1}||_1=n-1$. Therefore, the slack held by all nodes is $||\vec{s}_{r+1}||_1=1$ and the slack held by nodes in $L\subseteq V$ is $\sum_{x\in L}s_t[x]\leq 1$. We show now that indeed case 1 is a worst-case scenario. That is, in the complementary case 2 where some nodes have $d$ neighbors or more in one or more rounds, the slack is even smaller. To compare both scenarios we denote the slack for each round $t$, each node $x$, and each case $i$ as $s^{(i)}_t[x]$.

Assume that some node $x$ is the first one to have $d'>d-1$ neighbors. Let $1\leq t\leq r$ be the round of the first phase when this event happened.
We claim that $||\vec{s}_{t+1}^{(2)}||_1\leq ||\vec{s}_{t+1}^{(1)}||_1$. The reason is the following.
Given that more than $d-1$ potentials are received, node $x$ increases its potential to $1$ for the rest of the epoch (cf. Lines~\ref{leaderAlg}.\ref{leadertoomany} and~\ref{otherAlg}.\ref{othertoomany}). That is, the slack of $x$ is $s_{t+1}^{(2)}[x]\leq s_t^{(2)}[x]=s_t^{(1)}[x]$. 
%Then, counting only this effect, the overall slack $||\vec{s}_{t+1}^{(2)}||_1\leq ||\vec{s}_{t+1}^{(1)}||_1$. 
Additionally, the potential shared by $x$ with all neighbors during round $t$ is $d'\Phi_{t}[x]/d>\Phi_{t}[x](1-1/d)$ (cf. Lines~\ref{leaderAlg}.\ref{potupdate} and~\ref{otherAlg}.\ref{newpot}). That is, the potential shared by $x$ with neighbors in case 2 is more than the potential that $x$ would have shared in case 1. 
Then, combining both effects (the relative increase in potential of $x$ and its neighbors') the overall slack is $||\vec{s}_{t+1}^{(2)}||_1\leq ||\vec{s}_{t+1}^{(1)}||_1$. The same argument applies to all other nodes with $d$ or more neighbors in round $t$.

Additionally, for any round $t'$ of the first phase, such that $t<t'\leq r$, we have to additionally consider the case of a node $y$ that, although it does not receive more than $d-1$ potentials, it moves to alarm status because it has received an alarm in round $t'$. 
Then, notice that the potential of $y$ is $\Phi_{t'+1}[y]=1 \geq \Phi_{t'}[y]$, and it will stay in $1$ for the rest of the epoch (cf. Lines~\ref{leaderAlg}.\ref{alarminsecondleader} and~\ref{otherAlg}.\ref{alarminsecondother}). Therefore, the slack of $y$ is $s_{t+1}^{(2)}[y]\leq s_{t+1}^{(1)}[y]$. 

Combining all the effects studied over all rounds, the slack at the end of the first phase is $||\vec{s}_{r+1}^{(2)}||_1\leq ||\vec{s}_{r+1}^{(1)}||_1$. 
Given that $L\subseteq V$, it is $\sum_{x\in L}s_{r+1}^{(2)}[x] \leq ||\vec{s}_{r+1}^{(2)}||_1\leq ||\vec{s}_{r+1}^{(1)}||_1 \leq 1$ which completes the proof.

\hfill$\square$
\end{proof}

%%%%%%%%%%%%%%%%%%%%%%%%%%%%%%%%%%%%%%%%%%%%%%%%%%%%%%%%%%%%%%%%%%%%%%%%%%%

In our last lemma, we show that if $k^{1+\epsilon}<n$ the leader detects the error.

\begin{lemma}
\label{alarmsoon}
If $k^{1+\epsilon}<n$, $\epsilon>0$,
and $r\geq(4+2\epsilon- 2\ln(k^\epsilon-1)/\ln k)dk^2\ln k$,
within the following $k^{1+\epsilon}$ rounds after the first phase of the \name protocol, 
the leader has received an alarm message.
\end{lemma}
\begin{proof}

Using Theorem~\ref{koucky}, we know that after phase $1$ of $r\geq(4+2\epsilon- 2\ln(k^\epsilon-1)/\ln k)dk^2\ln k$ rounds, if $k=n$, the distribution is such that 
\begin{align*}
\left|\left|\vec{p}_{r+1} - \frac{\vec{I}}{k}\right|\right|_2^2 
&\leq \left(1-\frac{1}{d{\cal D}k}\right)^r\left|\left|\vec{p}_1 - \frac{\vec{I}}{k}\right|\right|_2^2\\
&\leq \exp\left(-\frac{r}{d{\cal D}k}\right)\nonumber\\
&\leq \exp\left(-\frac{(4+2\epsilon - 2\ln(k^\epsilon-1)/\ln k)dk^2\ln k}{d{\cal D}k}\right),\textrm{ given that $k=n>{\cal D}$,}\nonumber\\
&\leq \exp\left(-(4+2\epsilon- 2\ln(k^\epsilon-1)/\ln k)\ln k\right)\nonumber\\
&= 1/k^{4+2\epsilon- 2\ln(k^\epsilon-1)/\ln k}.\nonumber
\end{align*}

Given that for any node $j$, it is $(p_{r+1}[j] - 1/k)^2 \leq \left|\left|\vec{p}_{r+1} - \frac{\vec{I}}{k}\right|\right|_2^2$, we have that $(p_{r+1}[j]-1/k)^2 \leq 1/k^{4+2\epsilon- 2\ln(k^\epsilon-1)/\ln k}$. Hence, it is  $p_{r+1}[j] \leq 1/k + 1/k^{2+\epsilon- \ln(k^\epsilon-1)/\ln k}$ for any node $j$.
Moreover, if $k=n$ the total potential in the network would be $k-1$ (cf. Claim~\ref{conservation}) and no individual node should have potential larger than $(k-1)(1/k + 1/k^{2+\epsilon- \ln(k^\epsilon-1)/\ln k})$. We show that the latter is at most $\tau = 1-1/k^{1+\epsilon}$ as follows.
\begin{align*}
(k-1)(1/k + 1/k^{2+\epsilon- \ln(k^\epsilon-1)/\ln k}) &\leq 1-1/k^{1+\epsilon}\\
%1 + k/k^{2+\epsilon- \ln(k^\epsilon-1)/\ln k} - 1/k - 1/k^{2+\epsilon- \ln(k^\epsilon-1)/\ln k} &\leq 1-1/k^{1+\epsilon}\\
(k-1)/k^{2+\epsilon- \ln(k^\epsilon-1)/\ln k} &\leq (k^\epsilon -1)/k^{1+\epsilon}\\
%k^{2+\epsilon- \ln(k^\epsilon-1)/\ln k}/(k-1) &\geq k^{1+\epsilon}/(k^\epsilon -1)\\
k^{1 - \ln(k^\epsilon-1)/\ln k} &\geq (k-1)/(k^\epsilon -1)\\
\left(1 - \frac{\ln(k^\epsilon-1)}{\ln k}\right) \ln k &\geq \ln(k-1)-\ln(k^\epsilon -1)\\
\ln k &\geq \ln(k-1).
\end{align*}
And the latter is true for any $k>1$.

Consider a partition of the set of nodes $\{L,H\}$, 
where $L$ is the set of nodes with potential at most $\tau=1-1/k^{1+\epsilon}$ at the end of the first phase. 
At the end of the first phase, the size of $L$ is at most $k^{1+\epsilon}$ (cf. Lemma~\ref{unalarmed}), and the size of $H$ is at least $1$ because $n>k^{1+\epsilon}$. 
Thus, there is at least one node changing to alarm status in Line~\ref{otherAlg}.\ref{thresholdalarm} in round $1$ of phase $2$, and due to $1$-interval connectivity at least one new node moves from $L$ to $H$ in each of the following rounds.
Thus, the claim follows.

\hfill$\square$
\end{proof}

%%%%%%%%%%%%%%%%%%%%%%%%%%%%%%%%%%%%%%%%%%%%%%%%%%%%%%%%%%%%%%%%%%%%%%%%%%%

Based on the above lemmata, we establish our main result in the following theorem.

\begin{theorem}
\label{thm}
Given an \ADN with $n$  nodes, after running \name  for each estimate $k=2,3,\dots,n$ with parameters 
\begin{align*}
d &= k^{1+\epsilon},\\ 
p &= \left\lceil \frac{(2+\epsilon)k^{1+\epsilon}}{1-1/k}\ln k \right\rceil,\\ 
r &= \left\lceil \left(4+2\epsilon +\max\left\{0,- \frac{2\ln(k^\epsilon-1)}{\ln k}\right\}\right) dk^{2+2\epsilon}\ln k \right\rceil,\\
\tau &= 1-1/k^{1+\epsilon},
\end{align*}
where $\epsilon>0$,
all nodes stop after $\sum_{k=2}^n (pr+k)$ rounds of communication and output $n$. 
\end{theorem}
\begin{proof}
Notice that the above parameters fulfill the conditions of the previous lemmas. 

First we prove that \name is correct. To do so, it is enough to show that for each estimate $k<n$ the algorithm detects the error and moves to the next estimate, and that if otherwise $k=n$ the algorithm stops and outputs $k$. We consider three cases: $k=n$, $k<n\leq k^{1+\epsilon}$, and $k^{1+\epsilon}<n$, for a chosen value of $\epsilon>0$.

Assume first that $k<n\leq k^{1+\epsilon}$. Then, even if the leader does not receive an alarm during the execution, as shown in Lemma~\ref{ksquare}, at the end of the epoch in Line~\ref{leaderAlg}.\ref{range} the leader will detect that $\rho$ is out of range and will not change its status to done. Therefore, no other node will receive a termination message (loop in Line~\ref{leaderAlg}.\ref{leadernotification}), and all nodes will continue to the next epoch.

Assume now that $k^{1+\epsilon}<n$. Lemma~\ref{alarmsoon} shows that within the following $k^{1+\epsilon}$ rounds after the first phase the leader has received an alarm message, even if no node has more than $d-1$ neighbors during the execution and alarms due to this are not triggered. For the given value of $p$ and $k\geq 2$, the epoch has more than one phase. Therefore, within $k^{1+\epsilon}$ rounds into the second phase the leader will change to alarm status in Line~\ref{leaderAlg}.\ref{alarminsecondleader}, will not change its status to done later in this epoch, and no other node will receive a termination message. Hence, all nodes will continue to the next epoch. 

Finally, if $k=n$, Lemma~\ref{correct} shows that the accumulated potential $\rho$ will be $k-1-1/k\leq \rho\leq k-1$. Thus, in Line~\ref{leaderAlg}.\ref{range} the leader will change its status to done, and in the loop of Line~\ref{leaderAlg}.\ref{leadernotification} will inform all other nodes that the current estimate is correct. The number of iterations of such loop are enough due to $1$-interval connectivity. 

The claimed running time can be obtained by inspection of the algorithm, either for the leader or non-leader since they are synchronized.
Refer for instance to the leader algorithm in Figure~\ref{leaderAlg}.
The outer loop in Line~\ref{leaderAlg}.\ref{epochs} corresponds to each epoch with estimates $k=2,3,\dots, n$. For each epoch, Line~\ref{leaderAlg}.\ref{phases} starts a loop of $p$ phases followed by $k$ rounds in Line~\ref{leaderAlg}.\ref{leadernotification}. 
%The first phase includes $r$ rounds followed by additional $k^2$ rounds, whereas the remaining $p-1$ phases have only $r$ rounds. 
Each of the $p$ phases has $r$ rounds. 
Thus, the overal number of rounds is 
$\sum_{k=2}^n (pr+k)$.

\hfill$\square$
\end{proof}

%%%%%%%%%%%%%%%%%%%%%%%%%%%%%%%%%%%%%%%%%%%%%%%%%%%%%%%%%%%%%%%%%%%%%%%%%%%

Choosing $\epsilon = \log_k 2$, the following holds.

\begin{corollary}
The time complexity of \name is $O(n^5 \log^2 n)$.
\end{corollary}

\begin{align*}
\sum_{k=2}^n (pr+k)
&= \sum_{k=2}^n \left(\left\lceil \frac{(2+\epsilon)k^{2+\epsilon}}{k-1}\ln k \right\rceil \left\lceil \left(4+2\epsilon +\max\left\{0,- \frac{2\ln(k^\epsilon-1)}{\ln k}\right\}\right) k^{3+3\epsilon}\ln k \right\rceil + k \right)\\
&= \sum_{k=2}^n \left(\left\lceil \frac{2(2+\log_k 2)k^{2}}{k-1}\ln k \right\rceil \left\lceil \left(4+2\log_k 2 \right) 2^3k^{3}\ln k \right\rceil + k \right)\\
&\leq \sum_{k=2}^n \left(\left\lceil \frac{6k^{2}}{k-1}\ln k \right\rceil \left\lceil 48k^{3}\ln k \right\rceil + k \right)\\
&\in O(n^5 \log^2 n).
\end{align*}

%\textcolor{red}{--------- drafts follow}
%
%{\bf Alternative 1:}
%
%Consider the algorithm where the fraction shared corresponds to $k^2$ nodes. That is, Lemmas~\ref{correct} and~\ref{ksquare} still apply setting $d=k^2$.
%
%%Additionally, consider the following update of potential. For each round, each node transmits a fraction of its potential. After receiving from all neighbors, a node decreases its potential only by the fraction shared (as opposed to many fractions as before) and increases its potential by the sum of all fractions received divided by $k-1$, rounding down to $1$ if the new potential exceeds $1$. \mig{MM: is this matrix still stochastic? node $i$ decreases by a fraction of $1/k$ but neighbors are taking only $|N(i)|/k^2$, which may be less}
%
%Additionally, in the algorithm, after update, the potential is rounded to stay in $[\epsilon,1]$.
%
%No new alarms are included, besides the thresholds proved in the previous section.
%
%Claim: if we can prove that the total potential in the network never goes below $n-1$, the analysis of the previous section holds, properly adjusted to fit this modification of the algorithm.
%
%Negative answer: 
%Consider for instance a star with the leader in the center and $k=3$, that is, in each round nodes share $1/2$.
%After the first round, all nodes but the leader are left in $1/2$ but the leader in $1$, which is roughly $n/2$.
%In this case the leader knows that $k$ is wrong, but a similar example with a lollypop graph would have the same problem.

%removed special proof header

%!TEX root = ./KM_counting_SIAM.tex

\section{Extensions}
\label{s:extensions}

We argue that \name\ can be extended to compute the sum of values
stored in the nodes, and thus also the average (as it computes the number of
nodes $n$), and other functions.
Assume that each node of the \ADN\ initially stores a value, represented as
a sequence of bits. 
W.l.o.g. we could assume that the value stored at the leader is zero;
otherwise, the nodes could compute the sum of other initial values
(with the leader value set up to $0$), and later the leader could propagate its
actual initial value appended to the message ``done'' at the end of the execution
to be added to the computed sum of other nodes.

The modified \name\ 
%appends the potential to the beginning of the sequence. 
prepends the potential to the sequence.
Instead of sending potential by the original \name, each node transmits its current sequence (in which the potential stands in the first location).
Changes at each position of the sequence are done independly by the same algorithm as used for the potential, cf. Figures~\ref{leaderAlg} and~\ref{otherAlg}. 
Re-setting the values, in the beginning of each epoch, means
putting back the initial values of the sequence.
It means that the modified algorithm maintains potential in exactly the same
way as the original \name, regardless of the initial values. 
At the end of some epoch, with number corresponding to the number of
nodes $n$, all nodes terminate.
When it happens, each node recalls the sequence stored in it at the end
of the first phase of the epoch, multiplies the values stored at each position
of the sequence by the epoch number $n$,
and rounds each of the results to the closest integer; 
then it sums up the subsequent values multiplied by corresponding (consecutive) powers of $2$.
Note that such ``recalling'' could be easily implemented by storing
and maintaining the sequence after the first phase of each epoch.
%In particular,
%at some point the leader returns accurate number of nodes $n$.
%When this happens, instead of sending a simple message ``done'', it 
%rounds up the values stored at each position of the sequence,
%represents it as a binary vector,
%(adds its own initial value if different than $0$,)
%and keeps sending them for the next $n-1$ rounds,
%while during that time all other nodes relay this message upon receiving.
%Then all nodes terminate.

We argue that the computed value is the sum of the initial values.
It is enough to analyze how the modified algorithm processes values at
one position of the sequence, as positions are treated independently; 
therefore, w.l.o.g. we assume that each node has value $0$ or $1$ in the beginning.
Consider the last epoch before the leader sends the final sequence (in our
case, representing one value). In the beginning of the epoch,
the values are re-set to the original one, and manipulated independently
according to the rules in Figures~\ref{leaderAlg} and~\ref{otherAlg}.
Therefore, let us focus on the first phase of this epoch.
Since we already proved that the estimate %number 
of the last epoch is equal 
to the number of nodes, the value of $d$ in this epoch (and thus also 
in its first phase) 
is an upper bound 
on the node degree.
%, and t
Thus, the mass distribution scaled down by
the sum of the initial values behaves exactly the same as the probabilities 
of being at nodes in the corresponding round of the lazy random walk, with
parameter $d$ and starting from initial distribution equal to the
initial values divided by the sum. Since the length of the phase is set up to
guarantee that the distribution is close to the stationary uniform
within error $1/n$, 
and the sum of bits is not bigger than $n$,
at the end of the phase the value stored by each node is close to 
the sum (i.e., scaling factor) divided by $n$ by at most $1/n^4$  
(cf. Equation~\ref{distance}).
Therefore, after multiplying it by $n$, each node gets value of sum
within error of at most $1/n^3$, which after rounding will
give the integer equal to the value of the sum.

Once having the number $n$ and the sum, each node can
compute the average. 
As argued in~\cite{KDGgossip}, the capacity of computing the sum
of the input values makes possible the computation of more complex 
functions. Moreover, as opposed to~\cite{KDGgossip} where the computation
only converges, our approach outputs the exact sum. Therefore, the extension
to database queries that can be approximated using \emph{linear synopses}
\footnote{Additive functions on multisets, e.g. $f(A\cup B)=f(A)+f(B)$.} is straightforward.
Boolean functions $f:\{0,1\}^n\to\{0,1\}$, such as 
AND (sum  $= n$), OR (sum $>0$), and XOR (sum $=1$),
as well as their complementaries 
NAND (sum  $\neq n$), NOR (sum $= 0$), and XNOR (sum $\neq 1$),
can also be implemented having $n$ and the sum.
\dk{This applies also to other ``symmetric'' Boolean functions,
i.e., which do not depend on the order of variables,
as they could be computed based on computed sum of ones and $n$~\cite{KRANAKISboolean}.}
Maximum ($L_\infty$ norm) and minimum can be computed 
subsequently by flooding. That is, each node broadcasts the maximum 
and minimum input values seen so far. Due to $1$-interval connectivity 
within $n$ rounds all nodes have the answers.

\dk{Note that all these computations, including the \name, could be done
using only polynomial estimates of values, that is, with messages
of length $O(\log n)$, multiplied by the maximum number of coordinates
of any of the initial values. This could be also traded for time:
we could use only messages of length $O(\log n)$ with time increased
by the maximum number of coordinates of any initial value (which is still
polynomial in the size of the input,\footnote{% 
The input in this case is distributed among the nodes, and each node possesses 
at least one bit} 
which in this case is at least 
$n$ plus the maximum number of coordinates).
}

\section{Open Directions}

Straightway questions emerging from our work include existence of polynomial (in $n$)
lower bound and improvement of our upper bound.
One of the potential ways could be through investigating bi-directional relationships between random processes and computing algebraic functions in \ADN.
Extending the range of polynomially computable functions is another
intriguing future direction.
Finally, generalizing the model by not assuming connectivity in every round
or dropping assumption on synchrony could introduce even more challenging aspects of
communication and computation, including group communication and its impact
on the common knowledge about the system parameters.

\section*{Acknowledgments}
The authors would like to thank Michal Kouck{\`y} and Alessia Milani for useful discussions.

\bibliographystyle{abbrv}
\bibliography{Comprehensive_2010}

\end{document}